\documentclass[letterpaper,10pt, conference]{ieeeconf}

\IEEEoverridecommandlockouts
\usepackage{times}
\usepackage{setspace}
\usepackage{url}
\usepackage[utf8]{inputenc}
\usepackage[T1]{fontenc}
\usepackage{graphicx}		
\usepackage[caption=false,font=footnotesize]{subfig}
\usepackage{wrapfig}
\usepackage[format=plain,font=footnotesize,labelfont=bf,labelsep=period]{caption}
\usepackage{sidecap} 
\usepackage[export]{adjustbox}
\usepackage{float}

\usepackage{comment}

\usepackage{amsmath} 
\usepackage{amssymb}  
\usepackage{amsthm}
\usepackage{mathtools}
\usepackage[normalem]{ulem}
\usepackage{paralist}	
\usepackage[space]{grffile} 
\usepackage{color}
\let\labelindent\relax
\usepackage{enumitem}
\usepackage{bm}
\usepackage{cancel}
\usepackage{hhline}
\usepackage[c2 , nocomma]{optidef}
\usepackage{oubraces}
\usepackage{algorithmic}
\usepackage{textcomp}

\usepackage{moreverb, url}

\usepackage[ruled,vlined]{algorithm2e}
\usepackage{stackengine}
\usepackage{dirtytalk}

\usepackage{array}
\usepackage{verbatim}
\usepackage{upgreek}

\usepackage{amsfonts}
\usepackage{lipsum}
\usepackage{cite}
\usepackage{hyperref}

\newtheorem{theorem}{Theorem}

\theoremstyle{definition}
\newtheorem{definition}{Definition}
\theoremstyle{remark}

\theoremstyle{definition}
\theoremstyle{definition}

\newtheorem{problem}{Problem}

\DeclareMathOperator*{\diag}{diag}

\DeclareMathOperator*{\argmin}{arg\,min}

\newcommand{\R}{\mathbb{R}}
\newcommand{\C}{\mathcal{C}}

\definecolor{blue}{RGB}{38,38,134}
\definecolor{darkblue}{RGB}{0,0,102}
\definecolor{lightblue}{RGB}{77,77,148}

\definecolor{gold}{RGB}{234, 170, 0}
\definecolor{metallic_gold}{RGB}{139, 111, 78}

\newcommand{\software}[1]{\texttt{#1}}

\newcommand{\mb}[1]{\mathbf{ #1 }}

\usepackage{xfrac}

\usepackage{dblfloatfix}

\usepackage{pdfpages}

\usepackage{enumitem}

\begin{document}
\begin{spacing}{0.92}
\title{ \bf   
Safe Payload Transfer with Ship-Mounted Cranes: \\ A Robust Model Predictive Control Approach 
}

\author{Ersin Da\c{s}$^{1, 3}$$^{\star}$, William A. Welch$^{1}$$^{\star}$, Patrick Spieler$^{2}$$^{\star}$, Keenan Albee$^{2,4}$$^{\star}$, Aurelio Noca$^{1}$, Jeffrey Edlund$^{2}$, \\ Jonathan Becktor$^{2}$, Thomas Touma$^{1,5}$, Jessica Todd$^{1}$, Sriramya Bhamidipati$^{2}$, Stella Kombo$^{1}$, \\ Maira Saboia$^{2}$, Anna Sabel$^{2}$, Grace Lim$^{2}$, Rohan Thakker$^{2}$$^{\star}$, Amir Rahmani$^{2}$, and Joel W. Burdick$^{1}$
\thanks{$^{\star}$These authors contributed equally.}
\thanks{**This work was supported by DARPA under the LINC Program.}
\thanks{$^{1}$E. Da\c{s}, W. A. Welch, A. Noca, T. Touma, J. Todd, S. Kombo, J. W. Burdick are with the Department of Mechanical and Civil Engineering, California Institute of Technology, Pasadena, CA 91125, USA. ${\tt\small \{ersindas, jburdick \}@caltech.edu}$ }
\thanks{$^{2}$P. Spieler, K. Albee, J. Edlund, J. Becktor, S. Bhamidipati, M. Saboia, A. Sabel, G. Lim, R. Thakker, A. Rahmani are with the Jet Propulsion Laboratory, California Institute of Technology, Pasadena, CA 91109, USA. ${\tt\small amir.rahmani@jpl.nasa.gov}$ }
\thanks{$^{3}$E. Da\c{s} is also with the Department of Mechanical, Materials, and Aerospace Engineering, Illinois Institute of Technology, Chicago, IL 60616, USA. ${\tt\small edas2@illinoistech.edu}$}
\thanks{$^{4}$K. Albee is also with the University of Southern California, Los Angeles, CA 90089, USA. ${\tt\small kalbee@usc.edu}$ } 
\thanks{$^{5}$T. Touma is also with Stealth Labs, Pasadena, CA 91106, USA. ${\tt\small thomas@stealthlabshq.com}$ }
}

\maketitle
\pagestyle{plain}

\begin{abstract}
Ensuring safe real-time control of ship-mounted cranes in unstructured transportation environments requires handling multiple safety constraints while maintaining effective payload transfer performance. Unlike traditional crane systems, ship-mounted cranes are consistently subjected to significant external disturbances affecting underactuated crane dynamics due to the ship's dynamic motion response to harsh sea conditions, which can lead to robustness issues. To tackle these challenges, we propose a robust and safe model predictive control (MPC) framework and demonstrate it on a 5-DOF crane system, where a Stewart platform simulates the external disturbances that ocean surface motions would have on the supporting ship. The crane payload transfer operation must avoid obstacles and accurately place the payload within a designated target area. We use a robust zero-order control barrier function (R-ZOCBF)-based safety constraint in the nonlinear MPC to ensure safe payload positioning, while time-varying bounding boxes are utilized for collision avoidance. We introduce a new optimization-based online robustness parameter adaptation scheme to reduce the conservativeness of R-ZOCBFs. Experimental trials on a crane prototype demonstrate the overall performance of our safe control approach under significant perturbing motions of the crane base. While our focus is on crane-facilitated transfer, the methods more generally apply to safe robotically-assisted parts mating and parts insertion. 
\end{abstract}

\section{Introduction} 
The safe and accurate transfer or mating of parts and objects is a common objective in robotics applications.  While this paper focuses on problems related to ship-board cranes, the basic problem of safety during parts transfer is endemic in robotics. Crane systems are commonly used in various industries, including construction, manufacturing, ports, mining, and logistics, to lift and move heavy loads \cite{tordal2016, hoffman2021, bock2013real, wang2025event, ouyang2025, tian2024safety, mojallizadeh2023}. When these operations take place offshore, the same tasks must be performed on moving platforms, making safety and precision even more critical \cite{luan2025marine}. Specifically, ship-mounted cranes operating in dynamic marine environments, such as offshore logistics, ship-to-ship transfer, and subsea operations, are influenced by ocean wave-induced ship motion and wind conditions. These environmental disturbances, combined with the underactuated dynamics of the cranes and the swing of the payload, require a focus on safety constraints, including collision avoidance and compliance with operational limits. Therefore, ensuring safety is essential in addition to achieving performance goals, such as efficient payload transfer, highlighting the need for robust safety-critical control frameworks in these applications.

Various strategies have been proposed to enhance the robustness of control for underactuated ship-mounted cranes. These include nonlinear control \cite{wang2023ship}, adaptive sliding-mode control \cite{kim2019adaptive}, and neural network-based control \cite{yang2019neural}. Additionally, model predictive control (MPC)-based methods have been adapted for this control problem to enforce constraints and stabilize the system \cite{schubert2023, kimiaghalam2001, chen2023nmpc, cao2022lyapunov}. MPC schemes are particularly beneficial for integrating human operator control with the autonomous operation of cranes, as they can treat operator commands as references and apply the smallest corrective control input needed to satisfy operational constraints under actuator limitations \cite{giacomelli2019}. However, balancing robustness, safety, and performance continues to be a significant challenge for MPC. While ensuring safety is crucial, it is also important to avoid unnecessary conservativeness, as this can adversely affect performance, such as payload positioning accuracy. Conventional robust MPC methods commonly utilize tightened sets or chance constraints, which can be conservative in the presence of time-varying uncertainties. Furthermore, if the system dynamics are nonlinear, it is generally challenging to guarantee safety, which is often framed as forward set invariance, using MPC methods. 

Control barrier functions (CBFs) \cite{ames2017control} have emerged as a tool for synthesizing controllers that guarantee forward invariance of a given safe set. Incorporating CBFs into MPC constraints encodes forward invariance of the sets at each step and can improve closed-loop feasibility relative to standard state constraints \cite{zeng2021enhancing}. Discrete-time CBFs \cite{agrawal2017, ahmadi2019safe}, which extend continuous-time CBFs to certify safety at sampling instants, are combined with MPC \cite{zeng2021safety}. However, with discrete-time CBFs, safety is only enforced at sample times, which may result in transient inter-sample violations occurring between time steps. Moreover, if an appropriate discrete-time model is not selected for a continuous-time system, a high relative degree constraint of that system will continue to be represented as a high relative degree in the discrete-time model as well. To address these limitations, a notion of zero-order control barrier functions (ZOCBFs) for sampled-data systems was recently proposed in \cite{tan2025zero}. Although ZOCBFs are robust to sampling effects by design, their robustness against unmodeled system dynamics has not yet been explored. Additionally, ZOCBFs have not been tested in an experimental control system.  

In this paper, we propose a safe control framework for payload transfer using ship-mounted cranes within an operator-in-the-loop scheme. In the problem of payload positioning with a crane, complex geometric safety constraints arise from the target shape, the underactuated dynamics of the crane (especially as it is perturbed by the impact of ocean swells on its supporting ship), and the nature of the application. Our experimental application is motivated by the need to insert a cylindrical object into a tightly constrained cylindrical receptacle, without hitting objects on the deck, or hitting the surface of the deck beyond the lip of the receptacle.

We formulate a nonlinear MPC problem with a ZOCBF-based safety constraint to prevent collisions between the payload and the target set in the robot's operational space. To define a safe set for the payload around the target, we introduce a smooth safety function. We generate real-time computable time-varying bounding boxes for collision avoidance. Additionally, we introduce a sampling-based method for the online tuning of the robustifying parameter within the robust ZOCBF (R-ZOCBF) condition. This strategy allows the controller to adaptively ensure robustness against modeling uncertainties and disturbances while balancing the conservativeness of the control with the changing nature of the uncertainties. We manage this balance by optimizing the R-ZOCBF parameter using perturbations sampled from the current state. We demonstrate the effectiveness of our approach through hardware experiments on a 5-DOF crane, where we use a Stewart platform to simulate the external disturbances due to ocean wave-induced ship motion.

We organize this paper as follows: Section~\ref{sec:pre} introduces the formulation of our safety-critical crane control problem. Section~\ref{sec:main} presents our robust, safety-critical framework based on MPC with an online adaptive ZOCBF constraint. Section~\ref{sec:exp} discusses simulations and experiments, while Section~\ref{sec:conc} concludes the paper.

\section{Problem Statement} 
\label{sec:pre}
In this section, we first model the dynamics of our crane system. Next, we define a concept of safety using time-varying safety functions. Finally, we present the problem statement of our work. 

\textbf{Notation:} We use standard notation: ${\mathbb{R}}$ represents the set of real numbers, and $\mathbb{N}$ represents the set of natural numbers. The Euclidean norm of a vector is denoted by $\|\!\cdot\!\|$.  
An ${n \!\times\! m}$ zero matrix is denoted by ${\bf 0_{n \times m}}$. Similarly, ${\bf I_{n}}$ denotes the ${n \!\times\! n}$ identity matrix. For a vector ${v \!\in\! \mathbb{R}^n}$ and a diagonal matrix ${{W } \!\in\! \mathbb{R}^{n \times n}}$, we use the matrix-weighted norm ${\|v\|_{W}^2 \!\triangleq\! v^\top {W} v}$. The unit circle is denoted by $\mathbb{S}^1$.

\subsection{Underactuated 5-DOF Crane Dynamics}
\label{underactuated}

\begin{figure}[t]
    \centering
    \includegraphics[width=1\linewidth]{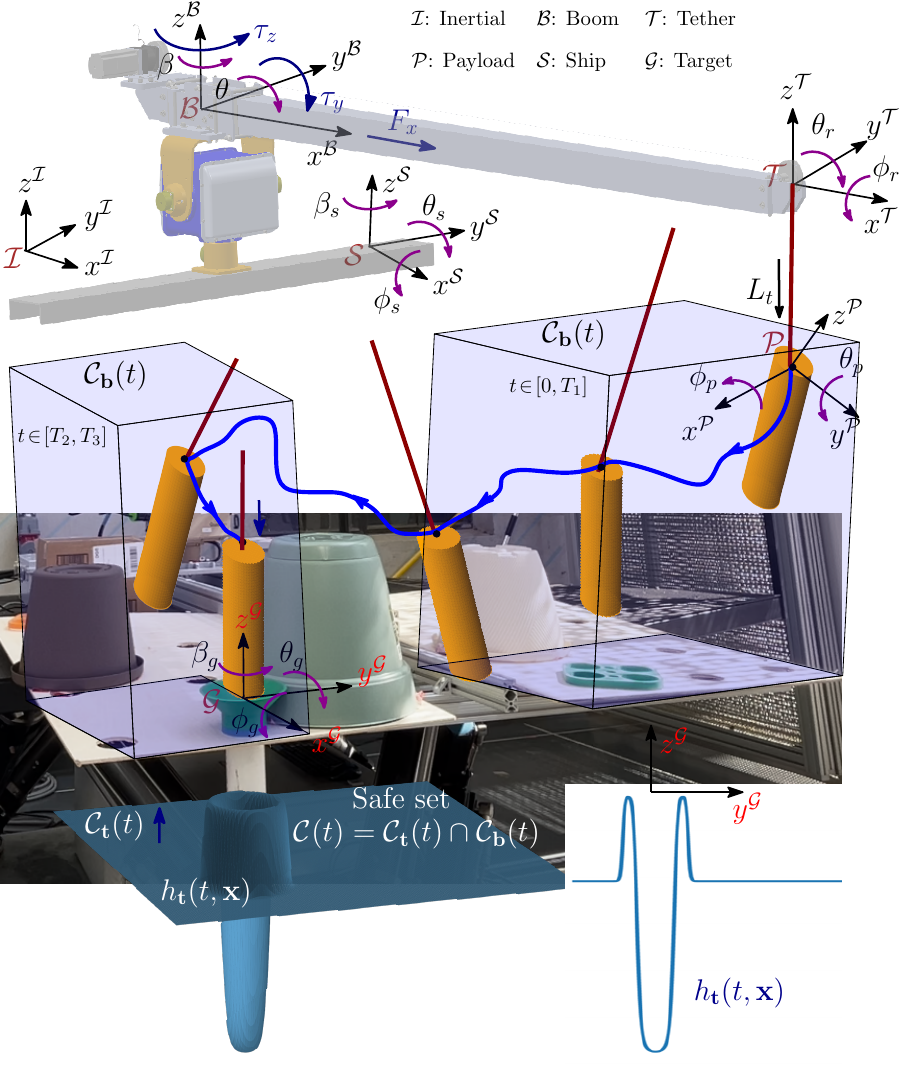}
    \caption{5-DOF ship-mounted crane with a task geometry and reference frames. Around the payload, time-varying collision-free \emph{bounding boxes} are depicted (Top). The smooth safety function for the target, ${h_{\mb t}}$, is visualized by its boundary, shown with a small display-only shift along ${x_{\mathcal{G}}}$ and ${-z_{\mathcal{G}}}$ to avoid occlusion (bottom-left). A cross-section of the zero-level set ${h_{\mb t}(t, \mb x(t)) \!=\! 0}$ (the boundary of the safe set $\C_{\mb t}(t)$, $\partial \C_{\mb t}(t)$) (bottom-right).} 
    \label{fig:crane}
     \vskip - 5mm
 \end{figure}

To obtain a nonlinear model of the crane system with the impact of ship motion on its dynamics, we use the Euler–Lagrange technique. In this work, we consider a 5-DOF electromechanical crane configuration illustrated in Fig.~\ref{fig:hardware}. We define the coordinate frames as depicted in Fig.~\ref{fig:crane}.

First, we define the generalized coordinates for the crane as ${\mb{q} \!\triangleq\! \left[ \beta~\theta~L_t~\phi_r~\theta_r~\phi_p~\theta_p \right]^\top}$, where ${\beta \!\in\! \mathbb{S}^1}$ is the crane yaw, ${\theta \!\in\! \mathbb{S}^1}$ is the boom pitch, ${L_t \!\in\! \mathbb{R}}$ is the tether length, and ${(\phi_r, \theta_r) \!\in\! \mathbb{S}^1 \!\times\! \mathbb{S}^1}$, ${(\phi_p, \theta_p) \!\in\! \mathbb{S}^1 \!\times\! \mathbb{S}^1}$ are the roll and pitch angles of the tether and the payload, respectively. The reference crane system has three electromechanical actuators, and the control input vector is ${ \tilde{\mb u} \!\triangleq\! \left[ \tau_y~\tau_z~F_x \right]^\top \!\in\! \mathbb{R}^3}$, where ${\tau_y}$ and ${\tau_z}$ are actuator torques about the pitch and yaw joints, respectively, and ${F_x}$ is the tether actuation force, see Fig.~\ref{fig:crane}. Note that the crane is a 5-DOF system, but we have a seven-dimensional generalized coordinates vector $\mb q$ because the last two components are the payload coordinates. 

We represent the measurable ship orientation, which induces external disturbances on the crane, with a time-varying vector: ${\mb d_s(t) \!\triangleq\! \left[ \beta_s~\theta_s~\phi_s \right]^\top}$, where ${(\beta_s, \theta_s, \phi_s) \!\in\! \mathbb{S}^1 \!\times\! \mathbb{S}^1 \!\times\! \mathbb{S}^1}$ are the yaw, pitch, and roll angles of the ship, respectively. We assume that ${\mb d_s(t), \dot{\mb d}_s(t), \ddot{\mb d}_s(t)}$ are measurable signals via onboard sensing, such as an inertial measurement unit, and might be provided by estimators. 

By assuming the tether is rigid, we use the geometry of the crane via forward kinematics to define the potential energy $V_e(\mb{q}, \mb d_s)$, and the kinetic energy $T_e(\mb{q}, \dot{\mb{q}}, \mb d_s, \dot{\mb d}_s)$ of the system. Then, the Euler–Lagrange equations with the Lagrangian: ${L_e(\mb{q}, \dot{\mb{q}}, \mb d_s, \dot{\mb d}_s) \!\triangleq\! T_e(\mb{q}, \dot{\mb{q}}, \mb d_s, \dot{\mb d}_s)-V_e(\mb{q}, \mb d_s)}$ are given by 
\begin{equation*}
\frac{d}{dt} \bigg ( \frac{\partial L_e(\mb{q}, \dot{\mb{q}}, \mb d_s, \dot{\mb d}_s)}{\partial \dot{\mb q}_i} \bigg ) - \frac{\partial L_e(\mb{q}, \dot{\mb{q}}, \mb d_s, \dot{\mb d}_s)}{\partial \mb q_i} = \tilde{\mb u}_i,
\end{equation*}
where ${\tilde{\mb u}_i}$ is the $i$-th component of ${\tilde{\mb u}}$, thus ${\tilde{\mb u}_i \!\equiv\! 0}$ for the \textit{passive coordinates}, which can only be controlled through actuated coordinates, similarly $\mb q_i$ denotes the $i$-th component of $\mb q$. Using these equations, we obtain the dynamics of the system in the form of
\begin{align}
\label{eq:model}
&\underbrace{\begin{bmatrix}
    D_{11}(\mb{q}) & D_{12}(\mb{q})\\
    D_{21}(\mb{q}) & D_{22}(\mb{q})
\end{bmatrix}}_{\triangleq D(\mb{q})}
\underbrace{\begin{bmatrix}
    \ddot{\mb{q}}_1 \nonumber \\
    \ddot{\mb{q}}_2 
\end{bmatrix}}_{\ddot{\mb{q}}} \\
&~~~~~~~~~~+ \underbrace{\begin{bmatrix}
    H_1(\mb{q}, \dot{\mb{q}}, \mb d_s, \dot{\mb d}_s, \ddot{\mb d}_s) \\
    H_2(\mb{q}, \dot{\mb{q}}, \mb d_s, \dot{\mb d}_s, \ddot{\mb d}_s)
\end{bmatrix}}_{\triangleq C(\mb{q}, \dot{\mb{q}}, \mb d_s, \dot{\mb d}_s) \dot{\mb{q}} + G(\mb{q}, \mb d_s,  \ddot{\mb d}_s)}  
= \underbrace{\begin{bmatrix}
    B_1(\mb q) \\ \bf 0_{4 \times 3}
\end{bmatrix}}_{\triangleq B(\mb{q})} \tilde{\mb u}, 
\end{align}
where ${\mb{q}_1 \!\triangleq\! \left[ \beta~\theta~L_t\right]^\top}$ are the \textit{actuated (active) coordinates}, while ${\mb{q}_2 \!\triangleq\! \left[\phi_r~\theta_r~\phi_p~\theta_p \right]^\top}$ denotes the passive coordinates, ${D(\mb{q}) \!\in\! \mathbb{R}^{7 \times 7}}$ is the inertia matrix, which is positive definite and symmetric, ${C(\mb{q}, \dot{\mb{q}}, \mb d_s, \dot{\mb d}_s) \!\in\! \mathbb{R}^{7 \times 7}}$ is the Coriolis matrix, ${G(\mb{q}, \mb d_s, \ddot{\mb d}_s) \!\in\! \mathbb{R}^7 }$ collects gravity and equivalent inertial forces induced by the measured ship accelerations, and ${B(\mb{q}) \!\in\! \mathbb{R}^{7 \times 3}}$ is the actuation matrix.

In practical applications, it is not always feasible to send torque commands directly to the actuators. Instead, feedback controllers typically send commands in the form of joint velocities ${\mb u \!\in\! \mathcal{U} \!\subset\! \mathbb{R}^3}$ for the actuated coordinates $\mb q_1$ as used in our hardware. We assume the inner joint velocity loops operate at a high bandwidth; therefore, we model the actuated coordinates' dynamics with a first-order velocity model:
\begin{equation}
\label{eq:actu}
\sigma_{a} \ddot{\mb{q}}_1 + {\dot{\mb{q}}_1} = {\mb{u}},
\end{equation} 
where ${\sigma_{a} \!\triangleq\! \diag(\sigma_{\beta}, \sigma_{\theta}, \sigma_{L})}$ with the actuator time constants ${\sigma_{\beta}, \sigma_{\theta}, \sigma_{L} \!>\! 0}$. Then, by substituting \eqref{eq:actu} into \eqref{eq:model} and setting ${\tilde{\mb u} \!\equiv\! 0}$, we obtain:
\begin{align*}
\underbrace{\begin{bmatrix}
    \!\sigma_{a} & \!\!\bf 0_{3 \times 4}\\
    D_{21}(\mb{q}) & D_{22}(\mb{q})
\end{bmatrix}}_{\triangleq \bar{D}(\mb{q})} \!\!
\begin{bmatrix}
    \ddot{\mb{q}}_1 \\
    \ddot{\mb{q}}_2 
\end{bmatrix} \!\!+\!\! \underbrace{\begin{bmatrix}
   {\dot{\mb{q}}_1} \\
    H_2(\mb{q}, \dot{\mb{q}}, \mb d_s, \dot{\mb d}_s, \ddot{\mb d}_s)
\end{bmatrix}}_{\triangleq \bar{H}(\mb{q}, \dot{\mb{q}}, \mb d_s, \dot{\mb d}_s, \ddot{\mb d}_s)}  
\!\!=\!\! \underbrace{\begin{bmatrix}
    \bf I_{3}  \\ \bf 0_{4 \times 3} 
\end{bmatrix}}_{\triangleq \bar{B}} \!\! {\mb u}. 
\end{align*}
We observe that $\bar{D}(\mb{q})$, a block lower-triangular matrix, is invertible because $D_{22}(\mb{q})$ and ${\sigma_{a}}$ are both positive definite. Therefore, the system dynamics can also be expressed in the form of a control-affine, time-varying nonlinear system:
\begin{align}
\label{eq:compact}
\!\underbrace{\begin{bmatrix}
    \dot{\mb{q}} \\
   \ddot{\mb{q}}
\end{bmatrix}}_{{ \triangleq \dot{\mb x}}} 
\!\!=\!\!
\underbrace{\begin{bmatrix}
    \dot{\mb{q}} \\
   - \bar{D}^{-1}(\mb{q}) ~ \! \bar{H}(\mb{q}, \dot{\mb{q}}, \mb d_s, \dot{\mb d}_s, \ddot{\mb d}_s)
\end{bmatrix}}_{\triangleq \mb{f}(t, \mb x)}  
\!+\!
\underbrace{\begin{bmatrix}
    \bf 0_{7 \times 3} \\
    \bar{D}^{-1}(\mb{q}) \bar{B}
\end{bmatrix}}_{\triangleq \mb{g}(\mb x)} \! \mb u,
\end{align}
where the drift dynamics ${\mb f \!:\! \mathbb{R}_{\geq 0} \!\times\! \mathcal{X} \!\to\! \R^{14}}$ is piecewise continuous in $t$ and locally Lipschitz continuous in $\mb x$, and time dependence enters via the measured ship motion, and the actuation matrix ${ \mb g \!:\! \mathcal{X} \!\to\! \R^{14 \times 3}}$ is locally Lipschitz continuous.

In practice, uncertainties and disturbances cannot be fully modeled. To represent the discrepancies between the actual model of the system and the derived model \eqref{eq:compact}, we utilize an uncertainty vector ${\mb{\delta}_e \!:\!  \mathbb{R}_{\geq 0} \!\to\! \mathbb{R}^{14} }$ such that 
\begin{equation}
\label{eq:sysun}
    \dot{\mb x}  = \mb{f}(t, \mb x) + \mb{g}(\mb x) \mb u + \mb{\delta}_e(t),
\end{equation}
where ${\mb{\delta}_e \!\in\! \mathcal{D}(t)}$ captures modeling inaccuracies and unmodeled disturbance inputs, such as wind. We assume that the uncertainty is bounded: ${\| \mb{\delta}_e(t)\|_{\infty} \!\leq\! \bar{\mb{\delta}}_e }$ for some ${\bar{\mb{\delta}}_e \!>\! 0}$.

Digital implementation of any controller introduces a sample-and-hold effect on the control:
\begin{equation}
\label{eq:ZOH}
\mb{u}(t) = \mb{u}_k, \quad \forall t \in [t_k, t_{k+1}),
\end{equation}
where $\mb{u}_k$ is constant between consecutive sampling instants $t_k$ and $ t_{k+1}$ with the time index ${k \!\in\! \mathbb{N}_0}$. We assume a uniform sampling strategy with a sampling interval $T$. As the system is piecewise locally Lipschitz continuous, it admits a unique trajectory. We can obtain a discrete-time nonlinear model of the continuous-time system \eqref{eq:compact} under the sample-and-hold strategy \eqref{eq:ZOH} as
\begin{equation*} 
{\mb x }_{k+1} \! \!=\! \mb x_k  + \!\! \int^{t_{k+1}}_{t_k} \!\!\! \left ( \mb{f}(\tau, \mb x(\tau) \!) \!+\! \mb{g}(\mb x(\tau)) \mb{u}_k \!\right )  d\tau \!\triangleq\! \mb{F}(t_k, \!\mb{x}_k, \!\mb{u}_k ),
\end{equation*}
where $\mb{F}$ is the \textit{state discrete map}. Similarly, we represent the \textit{exact state discrete map} under uncertainty $\mb{\delta}_e(t)$ with $\Phi(t_k, \!\mb{x}_k, \!\mb{u}_k, \mb{\delta}_e)$. And, ${\Phi_{\tau}(t_k, \!\mb{x}_k, \!\mb{u}_k, \mb{\delta}_e)}$ and $\mb{F}_{\tau}(t_k, \!\mb{x}_k, \!\mb{u}_k)$ are flows evaluated at time ${t_k  \!+\! \tau}$ for ${\tau \!\in\! [0~ T]}$.

\subsection{Characterization of Safety}
\label{sec:safety_f}

We encode each crane safety requirement with a time-varying, continuously differentiable constraint function. Let ${\mb y \!\in\! \mathbb{R}^y}$ denote an output coordinate of interest, such as payload position. Formally, a \emph{safety function}, ${h_{\mb y} (t, \mb y)}$, is defined as a continuous mapping ${h_{\mb y} \!:\! \mathbb{R}_{\geq 0} \!\times\! \mathbb{R}^y \!\to\! \mathbb{R}}$ that quantifies the safety for a particular output of the system, ${\mb y}$ at time ${t}$. The system is said to be safe at time ${t_0}$ with respect to the output coordinate ${\mb y}$ if ${h_{\mb y}(t_0, \mb y) \!\geq\! 0}$ and unsafe if ${ h_{\mb y}(t_0, \mb y) \!<\! 0}$. 

To relate safety functions to the system dynamics \eqref{eq:sysun}, we characterize safety through the notion of forward invariance. Let the time-varying safe set ${\C(t) }$ be the 0-superlevel set of a continuously differentiable function ${h\!:\! \mathbb{R}_{\geq 0} \!\times\! \mathcal{X} \!\to\! \mathbb{R}}$:
\begin{equation*}
    \C(t) \!\triangleq\! \left\{ \mb x \!\in\! \mathcal{X} \!\subseteq\! \mathbb{R}^n \!:\! h( t, \mb x) \!\geq\! 0 \right\}.
\end{equation*}
This set is \textit{forward invariant} if, for any initial condition ${\mb x(t_0) \!\in\! \C(t_0)}$, the solution $ \mb {x}(t) $ stays in ${\C(t)}$, ${\forall t \geq\! t_0}$. For a given time-varying feedback controller ${\mb{k} ( t, \mb x)}$, the closed-loop system ${\mb F_{\rm cl}(t, \mb x) \!\triangleq\!  \mb{f}(t, \mb x) \!+\! \mb{g}(\mb x) \mb{k} ( t, \mb x) \!+\! \mb{\delta}_e(t) }$ is {\em safe} with respect to the \textit{safe set} $\C(t)$ if $\C(t)$ is forward invariant.

In our safety-critical crane control application, we have multiple safety functions $h_i(t, \mb x)$ with their 0-superlevel sets $\C_i(t)$ for ${i \!\in\! \{1,\ldots, N_h\}}$. Since our goal is to satisfy all constraints simultaneously, we define the unified safe set as the intersection of these 0-superlevel sets. Equivalently, this intersection can be represented using the $\min$ function: 
\begin{equation}
\label{eq:safe_sets}
\C(t) \!\triangleq\! \bigcap_{i =1}^{N_h}  \C_i(t) \!\equiv\! \Big\{ \mb x \!\in\! \mathcal{X} \!: \min_{i \in \{1,\ldots, N_h\}} \left [ h_i(t, \mb x) \right ] \geq 0 \Big\}.
\end{equation}

\subsection{Safe Crane Control Problem Statement}

The objective of the control problem is to track a time-varying reference trajectory for the payload-bottom position ${ \mb r_{\mb{p}} \!:\! \mathbb{R}_{\ge 0} \!\to\! \mathbb{R}^3 }$ using the payload-bottom position output, expressed in an inertial frame ${\mb p_{\mb{p}} \!\triangleq\! [{\mb p}_{x}~{\mb p}_{y}~{\mb p}_{z}]^\top \!\in\! \mathbb{R}^3}$ and its velocity: ${\mb v_{\mb{p}}  \!\in\! \mathbb{R}^3}$, while minimizing tether and payload swing rates in the sense that
\begin{align}
\label{eq:target}
&\!\!\limsup_{t \to \infty} \!\left\| \mb p_{\mb{p}}  \!-\! \mb r_{\mb{p}} (t) \right\| \!\!\leq\! \epsilon,
\!~
\limsup_{t \to \infty} \!\left\| \mb v_{\mb{p}}  \!-\! \dot{\mb r}_{\mb{p}}(t) \right\| \!\!\leq\! \epsilon_v, \nonumber \\
&\big \| [\dot{\phi}_r(t)~ \dot{\theta}_r(t)]^\top \!\big \|^2  \leq \epsilon_r, \ \ \big \| [\dot{\phi}_p(t)~ \dot{\theta}_p(t)]^\top \big \|^2 \leq \epsilon_p.
\end{align}
The problem statement of this work follows:
\begin{problem}
\label{pro:first}
Given the system dynamics \eqref{eq:compact}, control input set $\mathcal{U}$, reference trajectory, and safe set \eqref{eq:safe_sets}, synthesize a robust, safe stabilizing controller to achieve the payload transfer objectives \eqref{eq:target} in the presence of model uncertainties and disturbances as defined in \eqref{eq:sysun}. 
\end{problem}
For a visual description of the problem, performance requirements, and safety constraints, see Fig.~\ref{fig:crane}: The payload follows a reference trajectory with minimal swing while staying inside the time-varying bounding boxes and above the target surface that shows the safe insertion boundary.

\section{Safe Payload Transfer}
\label{sec:main}
In this section, we begin by defining the safety functions for the crane control. Next, we introduce the concept of a time-varying robust ZOCBF. Finally, we propose a nonlinear MPC-based safe control design method with an online parameter-adaptive ZOCBF constraint to tackle Problem~\ref{pro:first}.  

\subsection{Crane Safety Functions}
\label{sec:safe_function}
In our safe control scenarios, we consider dynamic obstacles, which may move due to the motion of a ship, resulting in time-varying safe sets, see Fig.~\ref{fig:crane} and Fig.~\ref{fig:compose}.

\textbf{Time-varying bounding boxes for collision avoidance:} Our goal is to find the optimal free-space bounding box that fits in the point cloud around the payload. To keep the payload inside a safe rectangular region, we search for a time-varying box parameterized by lower and upper bounds: ${\bar{\mb p}_{x}^l(t), \bar{\mb p}_{x}^u(t), \bar{\mb p}_{y}^l(t), \bar{\mb p}_{y}^u(t), \bar{\mb p}_{z}^l(t), \bar{\mb p}_{z}^u(t) \!\in\! \mathbb{R}}$, where ${\bar{\mb p}_{(\cdot)}^l(t) \!<\! \bar{\mb p}_{(\cdot)}^u(t), \ \forall t \!\geq\! 0}$, are the lower and upper bounds for each ${x, y, z}$ coordinate. The bounding box provides a simple-to-compute convex approximation of the safe set originating from the payload frame. The associated safe set is given by
\begin{align}
\label{eq:bound_box}
&\C_{\mb b}(t) \!\triangleq\! \big \{ \mb x \!\in\! \mathcal{X} : \mb p_{\mb{p}}  \in  \mathcal{B}(t) \big \}, \\ \nonumber
&\mathcal{B}(t) \!\triangleq\! \big [  \bar{\mb p}_{x}^l(t), \bar{\mb p}_{x}^u(t) \big ] \!\times\! [ \bar{\mb p}_{y}^l(t), \bar{\mb p}_{y}^u(t) \big ] \!\times\! [ \bar{\mb p}_{z}^l(t), \bar{\mb p}_{z}^u(t)  \big ].
\end{align}
As detailed in \eqref{eq:safe_sets}, the set $\C_{\mb b}$ can also be characterized as the intersection of the 0-superlevel sets of six safety functions:
\begin{align*}
h_{1,2,3}( t, \mb x) \!=\! \mb p_{(\cdot)} \!-\! \bar{\mb p}_{(\cdot)}^l(t), \ h_{4,5,6}( t, \mb x) \!=\! \bar{\mb p}_{(\cdot)}^u(t) \!-\! \mb p_{(\cdot)}, 
\end{align*}
where ${(\cdot)}$ represents the ${x, y, z}$ coordinates of each vector.  

\textbf{A smooth safety function for target:} We assume that the safe target surface is cylindrical with a radius $r_t$. We need a smooth function that is positive when a cylindrical payload, with a radius $r_p$ such that ${r_p \!<\! r_t}$, is aligned with the target, continuing until the bottom of the payload makes contact with the target. If the payload is not aligned with the target, the function should prevent the payload from landing outside of the target. To enforce this constraint, we define a continuous safety function: 
\begin{align*}
h_{\mb t}(t,\! \mb x) \!=\! {\mb p}_{z} \!-\!
\bigg (\!
  \frac{\gamma_0}{1 \!+\! {\rm e}^{-\gamma_1  (\rho^2(t, \mb x) \!-\! \rho_1^2 ) } } 
  - \frac{\gamma_0 \!-\! c_0}{1 \!+ \! {\rm e}^{-\gamma_1 (\rho^2(t, \mb x) - \rho_2^2) }}
\!\!\bigg ),
\end{align*}
where ${\gamma_1 \!>\! 0}$ is a parameter to tune the steepness, ${c_0 \!>\! 0}$ tunes the height of the outer plateau below the peak, ${\gamma_0 \!>\! 0}$ is the peak amplitude, ${\rho_1 \!>\! \rho_2 \!>\! 0}$, ${\mb p_{\mb{t}}(t) \!\triangleq\! [{\mb t}_{x}~{\mb t}_{y}~{\mb t}_{z}]^\top \!\in\! \mathbb{R}^3}$ is the target position, and ${\rho^2(t,\mb x) \!\triangleq\! ({\mb p}_{x} - {\mb t}_{x})^2 + ({\mb p}_{y} - {\mb t}_{y})^2}$. When ${h_{\mb t}(t, \mb x) \!\geq\! 0}$, the payload is above the boundary of the target, i.e., it is in a safe set: ${\C_{\mb t}(t) \!\triangleq\! \left\{ \mb x \!\in\! \mathcal{X}  \!:\! h_{\mb t}( t, \mb x) \!\geq\! 0 \right\}}$. Therefore, the overall time-varying safe set is the intersection given by ${\C(t) \!=\! \C_{\mb t}(t) \cap \C_{\mb b}(t)}$, see Fig.~\ref{fig:crane}.

\subsection{Robust Zero-Order Control Barrier Functions} \label{sec:ZOCBF}
We consider dynamic obstacles in our safe control scenarios, which result in time-varying safe sets. Since our applied control framework is sampled-time, and the system has uncertain dynamics, we first introduce a notion of time-varying, robust zero-order control barrier functions (R-ZOCBFs), which extend time-invariant ZOCBFs \cite{tan2025zero}. We remark that with discrete-time CBFs \cite{agrawal2017}, safety is only enforced at sample times, which may result in transient inter-sample violations occurring between $t_k$ and $t_{k+1}$. ZOCBFs address inter-sample safety.

\begin{definition}[Robust Zero-Order Control Barrier Function]
\label{def:cbf}
Let ${\C(t) \!\subseteq\! \mathcal{X} }$ be the 0-superlevel set of a continuously differentiable function ${h\!:\! \mathbb{R}_{\geq 0} \!\times\! \mathcal{X} \!\to\! \mathbb{R}}$. The function $h$ is a \textit{robust zero-order control barrier function} for the sampled system with the exact state discrete map $\Phi(t_k, \!\mb{x}_k, \!\mb{u}_k, \mb{\delta}_e)$ if there exists a class-$\mathcal{K}$ function\footnote{A continuous function ${\alpha \!:\! [0, a ) \!\to\! \mathbb{R}_{\geq 0}}$, where ${a \!>\! 0}$, belongs to class-${\mathcal{K}}$ (${\alpha \!\in\! \mathcal{K}}$) if it is strictly monotonically increasing and ${\alpha(0) \!=\! 0}$.} ${\alpha \!\in\! \mathcal{K}}$ satisfying ${\alpha(r) \!\leq\! r, \ \forall r \!>\! 0}$, and $\delta$ such that 
${\forall \mb x_k \!\in\! \C(t_k)}$:
\begin{align*}
   \sup_{\mb u_k \in \mathcal{U}} \big [ h(t_{k+1}, \!\mb{F}(t_k, \!\mb{x}_k, \!\mb{u}_k )) \!-\! h(t_k, \!\mb x_k) \!-\! \delta(t_k, \!\mb{x}_k ) \big ] \\ 
   \geq -\alpha (h(t_k, \mb x_k)),
\end{align*}
and if ${\delta \!:\! \mathbb{R}_{\geq 0} \!\times\! \mathcal{X} \!\to\! \mathbb{R}_{\geq 0}}$ satisfies 
\begin{align}
 \label{eq:cbfdel}
  \delta(t_k, \!\mb{x}_k) \geq 
  \sup_{\mb{\delta}_e \in \mathcal{D}(t), \mb u \in \mathcal{U}} \big [\Theta(t_k, \!\mb{x}_k, \!\mb{u}, \mb{\delta}_e) \big ] \geq  0
\end{align}
with 
\begin{align*}
    \Theta \!\!\triangleq\! h(\!t_{k} \!+\! T, \!\mb{F}_{T}(t_k, \!\mb{x}_k, \!\mb{u} )\!) \!-\!\! \inf_{\tau \in [0,~ T]} \! h(t_k \!+\! \tau, \!\Phi_{\tau}(t_k, \!\mb{x}_k, \!\mb{u}, \mb{\delta}_e) \!).
\end{align*}
\end{definition}
We can derive formal safety guarantees using Definition~\ref{def:cbf}, with the help of the following theorem:
\begin{theorem} \label{thm: cbf}
If $h$ is an R-ZOCBF for the sampled system with maps $(\mb F,\Phi)$ with respect to $\C(t)$, then any feedback controller ${\mb k   \!:\! \mathbb{R}_{\ge 0}  \!\times\! \mathcal{X} \!\to\! \mathcal{U}}$, ${\mb u \!=\! \mb k(t, \mb x)}$, satisfying 
\begin{align} \label{eq:cbf_condition}
    h(t_{k+1}, \mb{F}(t_k, \mb{x}_k, \!\mb{u}_k )) - h(t_k, \mb x_k) \!-\! \delta(t_k, \!\mb{x}_k) \nonumber \\ 
    \geq -\alpha (h(t_k, \mb x_k))
\end{align}
for all ${\mb x_k \!\in\! \C(t_k)}$ renders the set $\C(\cdot)$ forward invariant for all admissible uncertainties ${\mb{\delta}_e \!\in\! \mathcal{D}(t)}$.
\end{theorem}
\begin{proof}
\label{pr:DCBF}
To prove the statement of the theorem, for any ${k \!\in\! \mathbb{N}_0}$ assume that ${\mb x_k \!\in\! \C(t_k)}$, i.e., ${h(t_k,\mb x_k) \!\ge\! 0}$. From condition \eqref{eq:cbf_condition} with ${\Phi(t_k, \!\mb{x}_k, \!\mb{u}_k, \mb{\delta}_e)}$, and \eqref{eq:cbfdel} we have
\begin{align*}
&h(t_{k+1}, \Phi(t_k, \!\mb{x}_k, \!\mb{u}_k, \mb{\delta}_e)) \\
&\geq h(t_{k+1}, \mb{F}(t_k, \mb{x}_k, \!\mb{u}_k )) - \delta(t_k, \!\mb{x}_k)   \\
&\geq h(t_k,\mb x_k)  \!-\! \alpha ( h(t_k, \mb x_k) ) \!\triangleq\! \bar{\alpha} ( h(t_k, \mb x_k) ),
\end{align*}
where ${\bar{\alpha} (r) \!\triangleq\! r \!-\! \alpha(r)}$. Since ${\alpha \!\in\! \mathcal{K}}$, and ${\alpha(r) \! \leq \! r}$ for all ${r \!\geq\! 0}$, we have ${\bar{\alpha}(r) \!\geq\! 0}$ for all ${r \!\geq\! 0}$, and ${\bar{\alpha}(0) \!=\! 0}$. Therefore, it follows that ${h(t_{k+1}, \Phi(t_k, \!\mb{x}_k, \!\mb{u}_k, \mb{\delta}_e)) \!\geq\! 0}$, which implies ${\Phi(t_k, \!\mb{x}_k, \!\mb{u}_k, \mb{\delta}_e) \!\in\! \C(t_{k+1})}$. Since ${\mb x_0 \!\in\! \C(t_0)}$ by assumption, by induction ${\mb x_k \!\in\! \C(t_k)}$ for all ${k \!\in\! \mathbb{N}_0}$; hence $\C(\cdot)$ is forward invariant under any feedback controller that enforces \eqref{eq:cbf_condition} at each time step $t_k$.
\end{proof}

\subsection{Safety-Critical Model Predictive Control}
In order to synthesize the safe controller, we formulate a continuous-time nonlinear MPC problem:
\begin{align*}
\begin{array}{l}
\!\!{{{\mb x}^*(t)},\! {\mb u^*(t)} \!=\!  \ }
\displaystyle \! \argmin_{ \mb{x}(\cdot), \ \mb u(\cdot) } \!\!~ {\mathcal{J}_{T}(T_f, \!\mb x (T_f)\!) \!+\!\!\! \int_{t = 0}^{T_f}  \!\!\! \mathcal{J}(t, \!\mb x(t), \!\mb u(t)) dt  }   \\ [4mm]
~~~~~~~~~~~~~~~~~\!\textrm{s.t.} \\ [2mm]
\!\dot{\mb x}(t) \!=\! \mb{f}(t, \mb x(t)) \!+\! \mb{g}(\mb x(t)) \mb u(t), \ \mb x(0) \!=\! \mb x_0, \  \forall t \!\in\! [0, T_f]\\ [2mm]
\!h_{\mb t}(t\!+\!T, \mb{F}(t, \mb{x}(t), \mb{u}(t) )) \!-\! h_{\mb t}(t, \mb x(t)) \!-\! \mb{\delta}_{\mb t}(t, \mb{x}(t)) \\[1mm]
~~~~~~~~~~~~~~~~~~~~~~~~~\!\geq\! -\alpha (h_{\mb t}(t, \mb x(t))), \ \forall t \!\in\! [0, T_f \!-\!T] \\ [2mm]
\!{h}_i(t, \mb x(t)) \!\geq\! 0, i \!\in\! \{1,\ldots, 6\}, \ \forall t \!\in\! [0, T_f] \\ [2mm]
\!\mb u(t) \!\in\! \mathcal{U}, \ \forall t \!\in\! [0, T_f] \\ [2mm]
\!\mb x(t) \!\in\! \mathcal{X}, \ \forall t \!\in\! [0, T_f],
\end{array}
\end{align*}
where ${\mb u^*(t)}$ and ${\mb x^*(t)}$ are the respective optimal input and state signals, $T_f$ is the horizon. The objective function has two terms: \textit{the stage cost} ${\mathcal{J} \!:\! \mathbb{R}_{\geq 0} \!\times\! \mathcal{X} \!\times\! \mathcal{U} \!\to\! \mathbb{R}_{\geq 0}}$ and \textit{the terminal cost} ${\mathcal{J}_T \!:\! \mathbb{R}_{\geq 0} \!\times\! \mathcal{X} \!\to\! \mathbb{R}_{\geq 0}}$. We formulate $\mathcal{J}$ as a summation of six different functions balancing input use and payload trajectory tracking and swing rate: 
\begin{align*}
    &\mathcal{J}(t, \mb x(t), \mb u(t)) = \big \| \mb u(t) \big \|_{W_{1}}^2 + \big \| \mb u(t) - {\mb u}_m(t) \big \|_{W_{2}}^2  \\
    &~~~~~+ \big \| \mb p_{\mb{p}} (t, \mb x(t)) - \mb r_{\mb{p}} (t) \big \|_{W_{3}}^2  + \big \| \mb v_{\mb{p}} (t, \mb x(t)) - \dot{\mb r}_{\mb{p}}(t) \big \|_{W_{4}}^2 
   \\
  &~~~~~+ \big \| [\dot{\phi}_p(t)~ \dot{\theta}_p(t)]^\top \big \|_{W_{5}}^2 + \big \| [\dot{\phi}_r(t)~ \dot{\theta}_r(t)]^\top \!\big \|_{W_{6}}^2, 
\end{align*}
where ${W_i \succeq 0}$, ${i \!\in\! \{1,\ldots, 6\}}$,  are diagonal weighting matrices, and ${{\mb u}_m(t) \!\in\! \mathcal{U}}$ represents the actual measured joint velocities. Similarly, we formulate the terminal cost as
\begin{align*}
    &\mathcal{J}_T(t, \mb x(T_f)) =  \\ 
    &  \big \| \mb p_{\mb{p}} (T_f, \!\mb x(T_f)) \!-\! \mb r_{\mb{p}} (T_f) \big \|_{W_{3}}^2 \!\!+\! \big \| \mb v_{\mb{p}} (T_f, \!\mb x(T_f)) \!-\! \dot{\mb r}_{\mb{p}}(T_f) \big \|_{W_{4}}^2   
   \\
  & + \big \| [\dot{\phi}_p(T_f)~ \dot{\theta}_p(T_f)]^\top \big \|_{W_{5}}^2 + \big \| [\dot{\phi}_r(T_f)~ \dot{\theta}_r(T_f)]^\top \big \|_{W_{6}}^2. 
\end{align*}

In the cost function $\mathcal{J}$, the last four terms represent the objectives outlined in \eqref{eq:target} and stated in Problem~\ref{pro:first}. To penalize the control input, the MPC utilizes ${\big \| \mb u(t) \big \|_{W_{1}}^2}$, which denotes the weighted norm of the control vector. The term ${\big \| \mb u(t) - {\mb u}_m(t) \big \|_{W_{2}}^2}$ penalizes deviations between commanded and measured joint velocities, helping to minimize excessive fluctuations in the control input.

In our implementation, we utilize the \software{acados} software package \cite{acados}. 
In particular, we employ the SQP real-time iteration (RTI) framework \cite{diehl2002real} with the HPIPM partial-condensing QP solver. We set a nonlinear MPC horizon of ${T_f \!=\! 1 \!~s}$ discretized into ${N_T \!=\! 30}$ intervals in \software{acados}.

Uncertainty and external disturbances are unmeasurable dynamics in general; therefore, $\mb{\delta}_e$ in \eqref{eq:sysun} is unknown. Ignoring the effects of $\mb{\delta}_e$ on the safety condition in \eqref{eq:cbf_condition} may lead to unsafe behavior. Inspired by uncertainty estimation methods \cite{chen2015disturbance}, which involve filtering the discrepancy between the measured states and nominal dynamics, we can observe the effects of uncertainty on each state. We utilize these effects to represent $\mb{\delta}_e$, whose effect on ${h_{\mb t}}$ is denoted with ${\mb{\delta}_{\mb t} \!\in\! \mathbb{R}_{\geq 0}}$ for all ${t \!\geq\! 0}$ in the MPC problem, with an upper bound in the R-ZOCBF condition at each time step $t$.


\begin{figure*}[t]
  \centering
  \setlength{\tabcolsep}{2pt}
  \renewcommand{\arraystretch}{1.0}
  \begin{tabular}{ccccc}
    \includegraphics[width=0.19\textwidth]{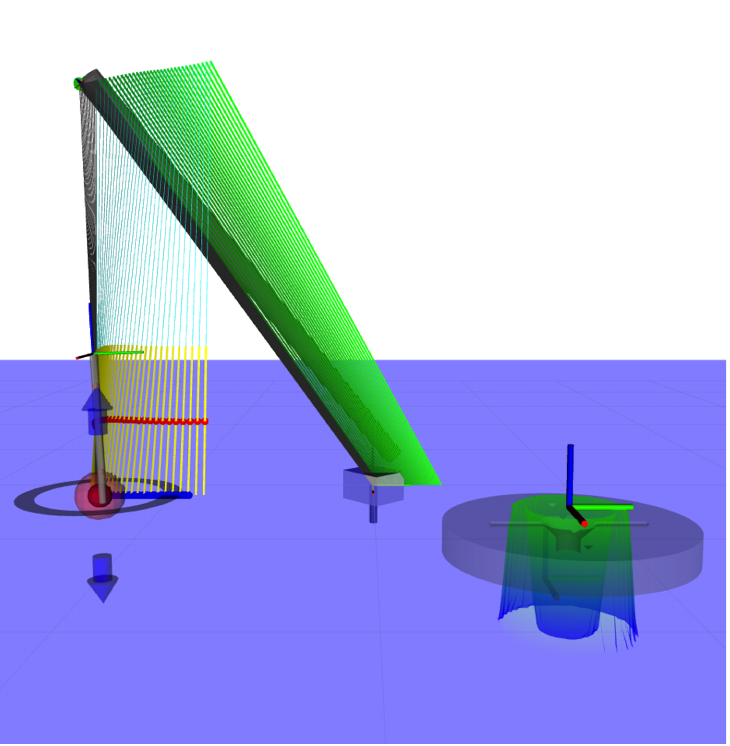} &
    \includegraphics[width=0.19\textwidth]{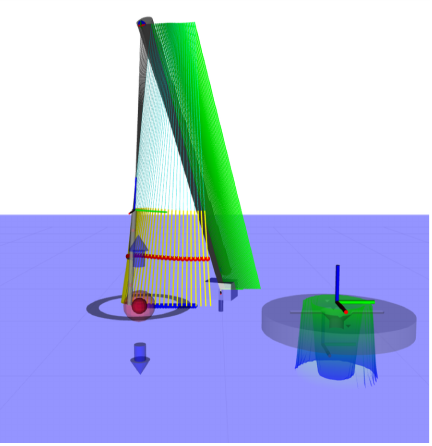} &
    \includegraphics[width=0.19\textwidth]{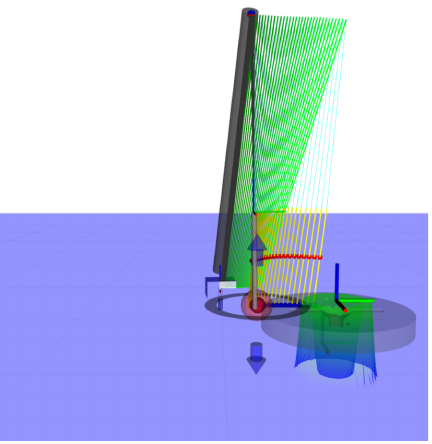} &
    \includegraphics[width=0.19\textwidth]{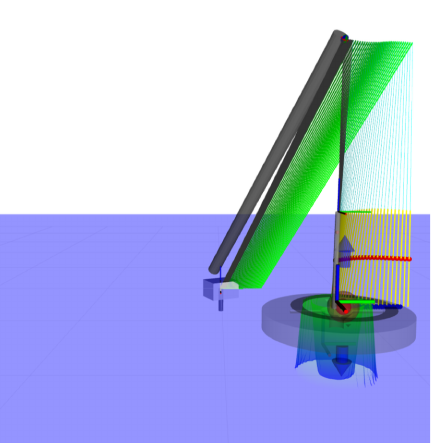} &
    \includegraphics[width=0.19\textwidth]{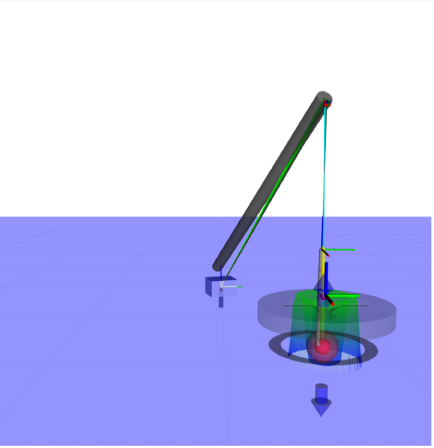} \\
    \footnotesize $t=0$ & \footnotesize $t=2$ & \footnotesize $t=4$ & \footnotesize $t=6$ & \footnotesize $t=10$
  \end{tabular}
  \caption{Sequence of the platform-mounted crane performing payload insertion. Snapshots at \(t\in\{0, 2, 4, 6, 10\}\ \mathrm{s} \) show approach, alignment and insertion of the payload while the base undergoes sinusoidal motion. A green-to-blue colored overlay shows the safety function constraint near the target. The MPC-calculated trajectories of the boom and payload are visualized using green segments for the boom and yellow segments for the payload.}
  \label{fig:sim_snapshots}
   \vskip - 5mm
\end{figure*}

\begin{figure*}[t]
  \centering

  \subfloat[Simulation results with nominal safety function. \label{fig:nominal_payload_safety_sim}]{
    \begin{minipage}{0.48\linewidth}
      \centering
      \includegraphics[width=0.49\linewidth]{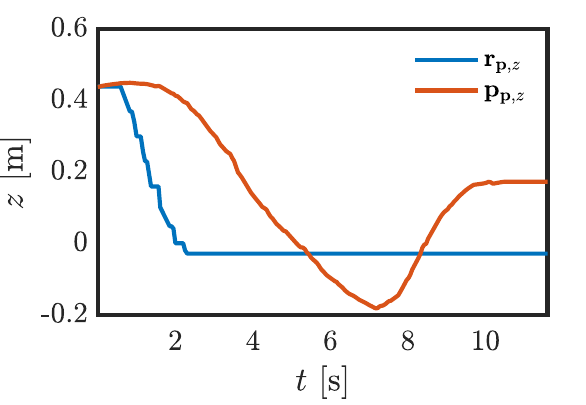}\hfill
      \includegraphics[width=0.49\linewidth]{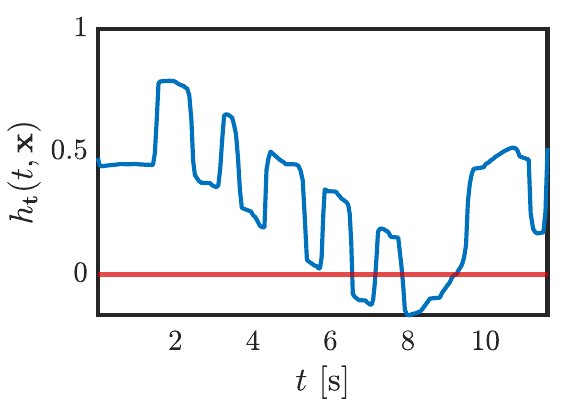}
    \end{minipage}
  }
  \hfill
  \subfloat[Simulation results with robust safety constraint. \label{fig:robust_payload_safety_sim}]{
    \begin{minipage}{0.48\linewidth}
      \centering
      \includegraphics[width=0.49\linewidth]{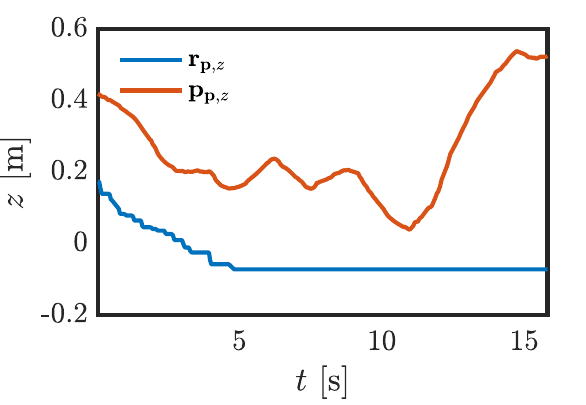}\hfill
      \includegraphics[width=0.49\linewidth]{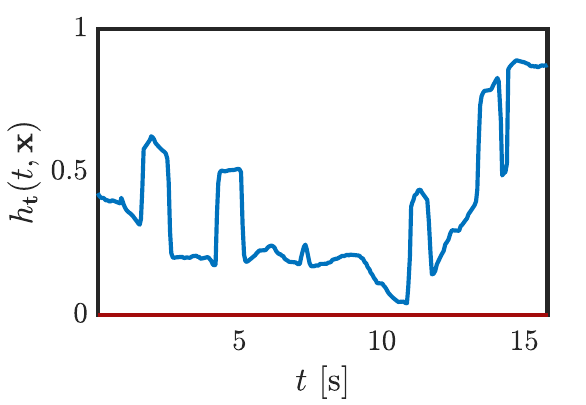}
    \end{minipage}
  }

  \caption{Comparison of nominal and robust safety configurations across simulation runs. Each subfigure contains the vertical payload tracking $(\mathbf{r}_{\mathbf{p},z}, \mathbf{p}_{\mathbf{p},z} \ \text{vs.\ time})$ on the left and the corresponding safety function $h_{\mb t}(t,\mathbf{x}(t))$ on the right. Safety is violated when $h_{\mb t}(t, \mb x(t))<0$. Our robust MPC with R-ZOCBF constraint maintains safety, while the MPC with the nominal safety function violates safety. 
}
  \label{fig:simulation_payload_safety_panel}
   \vskip - 5mm
\end{figure*}

We assume that a bound for the effect of uncertainties on the states: ${\| {\mb x}  \!-\!  \hat{\mb x} \| \!\leq\! \Delta}$, where $\hat{\mb x}$ is the state evaluated with uncertainty, i.e., actual state, for some ${\Delta \geq 0}$ is given. In practice, system identification methods, or data-driven uncertainty predictors can provide an estimation of $\Delta$ along with an empirical prediction bound. 

To select the parameter $\mb{\delta}_{\mb t}$ online, we define an optimization problem at each time step $t_k$ in ${t \!\in\! [0, T_f]}$:
\begin{align*}
\begin{array}{l}
    \mb{\delta}_{\mb t} (t, {\mb x}) = \displaystyle \! \argmin_{ \mb{\delta}_{\mb t} \in \mathbb{R}_{\geq 0}} \!\! ~~~\mb{\delta}_{\mb t} 
    \\ [2mm]
~~~~~~~~~~~~\textrm{s.t.} \\
\mb u(t) \!\in\! \mathcal{U}, \  \!{\mb x}(t) \!\in\! \mathcal{X} \\ [2mm]
h_{\mb t}(t\!+\!T, \mb{F}(t, \!{\mb x}(t), \!\mb{u}(t) )\!) \!-\! 
h_{\mb t}(t_k \!+\! \tau, \!\Phi_{\tau}(t, \!{\mb x}, \!\mb{u}, \mb{\delta}_e) \!)
\!-\! \mb{\delta}_{\mb t} \\ [1mm]
~~~~~~~~~\geq -\alpha (h_{\mb t}(t, {\mb x}(t))), ~~\forall \tau \!\in\! [0,~T], ~~\forall \mb{\delta}_e \in \mathcal{D}(t).  
\end{array}
\end{align*}
This optimization problem requires evaluating the dynamics with uncertainty, which is unknown; furthermore, it assumes that the value of $\mb u$, which is a decision variable in the MPC problem, is given; therefore, it is not possible to solve this problem. Additionally, because of the requirement ${\forall \tau \!\in\! [0, T]}$ in the last constraint of the optimization problem, this constraint represents an infinite number of constraints. To address these issues, we approximate $\hat{\mb x}$ by sampling ${N \!\in\! \mathbb{N}}$ points, $\{\tilde{\mb x}_1, \dots, \tilde{\mb x}_N \}$, from the closed ball of radius ${\Delta}$ centered at ${{\mb x}}$ in $\R^n$, $B_\Delta({\mb x})$. We also discretize the constraint for ${\tau \!\in\! \{ 0, \Delta_{\tau}, \ldots, T\}}$, where $\Delta_{\tau}$ is a discretization step satisfying ${T/\Delta_{\tau} \!\in\! \mathbb{N}}$. And, we approximate $\mb u$ with the input $\mb u^*_{k-1}$ applied during the time interval $[t_{k-1}, t_k)$ by MPC. Finally, we select ${\mb{\delta}_{\mb t}}$ by solving a sampling-based optimization problem:
\begin{align*}
\begin{array}{l}
    \mb{\delta}_{\mb t} (t, {\mb x}) = \displaystyle \! \argmin_{ \mb{\delta}_{\mb t} \in \mathbb{R}_{\geq 0}} \!\! ~~~\mb{\delta}_{\mb t} 
    \\ [1mm]
~~~~~~~~~~~~~~~~~\!\textrm{s.t.} ~~~ \!\tilde{\mb x}(t) \!\in\! \{\tilde{\mb x}_1, \dots, \tilde{\mb x}_N \}, \tilde{\mb x}_i \!\in\! B_\Delta({\mb x}) \\ [2mm]
\!h_{\mb t}(t\!+\!T, \mb{F}(t, {\mb x}(t), \mb u^*_{k-1} )) \!-\! h_{\mb t}(t \!+\! \tau, \!\mb{F}_{\tau}(t, \!\tilde{\mb x}, \!\mb u^*_{k-1}) ) \!-\! \mb{\delta}_{\mb t} \\[1mm]
~~~~~~~~~~~~~~~~~~~~\geq -\alpha (h_{\mb t}(t, {\mb x}(t))), \ \tau \!\in\! \{ 0, \Delta_{\tau}, \ldots, T\}. 
\end{array}
\end{align*}
We remark that this optimization problem has a single decision variable, and both $ \tilde{\mb x}$ and $\tau$ are taken from finite grids; the program simplifies to a one-dimensional linear check: evaluate each sampled constraint and select the most restrictive bound.

\section{Simulations and Experiments} 
\label{sec:exp}
In this section, we evaluate the performance of our robust MPC on the crane system shown in Fig.~\ref{fig:crane}, using both simulations and hardware experiments.


\subsection{Simulations}

Simulations were conducted in RViz on a workstation running Ubuntu 22.04, equipped with an Intel Core i9 processor. The system was modeled symbolically as described in Section~\ref{sec:pre} and the optimal control problem was solved using \software{acados} \cite{acados}, which is a tool for nonlinear optimization.

In our simulation, the control objective is to cause the payload to follow a desired insertion trajectory while avoiding collisions with the surrounding environment and the target structure, as shown in Fig.~\ref{fig:sim_snapshots}. In this simulation, the base translates laterally with sinusoidal motion, resulting in a lateral displacement of $\pm 0.05~\mathrm{m}$ (peak-to-peak $(0.10~\mathrm{m}$) with one full cycle every second.

To enforce safety, the MPC applies an R-ZOCBF constraint to the target safety function $h_{\mb t}$ (Sections~\ref{sec:safe_function} and~\ref{sec:ZOCBF}). This constraint prevents payload contact with the environment and target walls while allowing the payload to descend below a clearance threshold when already within the target. 

To evaluate robustness, we model the effect of external disturbances and state estimation errors as bounded uncertainties on each state variable. The uncertainty vector represents model inaccuracies and measurement errors in the hardware that the controller must tolerate. Angular uncertainty is bounded by $2^\circ$ payload roll and pitch angles, as well as base yaw and pitch. Angular velocity estimates are expected to deviate up to ${4~^\circ/\text{s}}$ for the tether and payload. Angular velocity estimates for the base yaw and pitch are bounded by ${3~^\circ/\text{s}}$. The tether length estimate can deviate by as much as ${0.04~\text{m}}$, with corresponding rate errors up to ${0.08~\text{m/s}}$. The payload’s position is subject to uncertainty of up to $0.03~\text{m}$ per axis, with velocity errors allowed up to ${0.09~\text{m/s}}$ per axis. The target’s position estimate errors are bounded by ${0.03~\text{m}}$ per axis. 

We first analyze the performance of our MPC controller using the baseline safety constraint ${h_{\mb t}(t,\mb x) \!\geq\! 0}$ under dynamic base motion. Fig.~\ref{fig:nominal_payload_safety_sim} shows the vertical trajectory and corresponding safety constraint violations, demonstrating that the baseline constraint is insufficient to achieve safety. In the absence of our robust safety constraint, the payload violates our safety condition, resulting in a collision with the target or the environment. We then perform the same simulation using our robust safety constraint. Fig.~\ref{fig:robust_payload_safety_sim} shows the vertical trajectory and safety constraint satisfaction for this scenario, indicating safe operation and robustness to the base disturbance and model uncertainties.

These simulations demonstrate that integrating an R-ZOCBF constraint into MPC significantly enhances safety for crane operations under uncertainty. The robust MPC prevents payload collisions while minimizing additional control effort. These results highlight the effectiveness of our approach for real-time, safe payload transfer.

\subsection{Experimental Results}
\label{sec:hardware}

\begin{figure}[t]
    \centering
    \vspace{0.22cm}
    \includegraphics[width=1\linewidth]
    {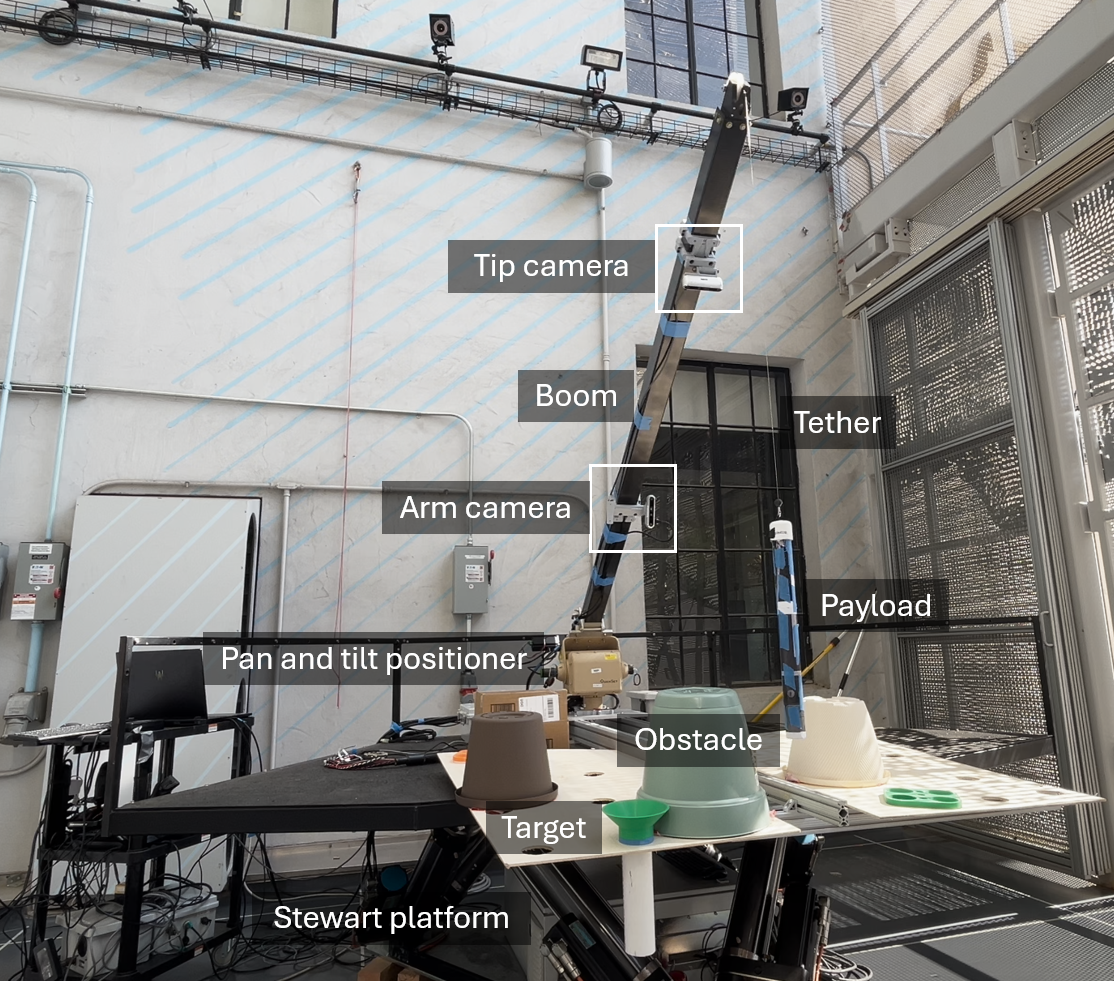}%
    \caption{Experimental setup of the reference crane system and the Stewart platform. 
    }
    \label{fig:hardware}
    \vskip - 5mm
\end{figure}

\begin{figure*}[t]
    \centering
    \includegraphics[width=1\linewidth]{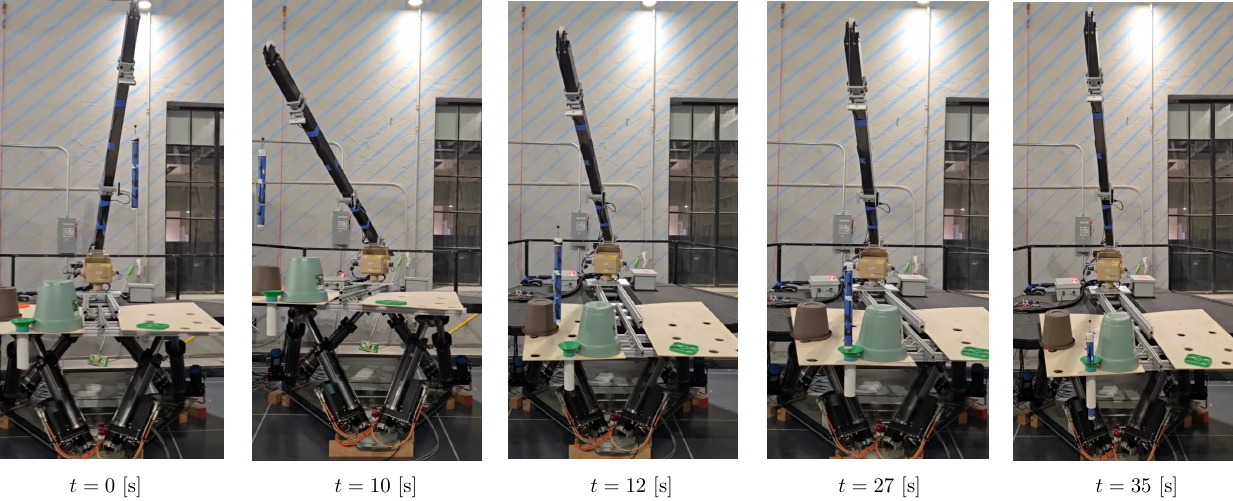}
    \caption{Hardware demonstration. The proposed robust MPC-based safe control framework ensures safe payload transfer to a target under external disturbances simulated by the Stewart platform and model uncertainties. } 
    \label{fig:compose}
    \vskip - 5mm
 \end{figure*}

\begin{figure*}[t]
  \centering

  \subfloat[Hardware results with the nominal safety function.\label{fig:nominal_payload_safety}]{
    \begin{minipage}{0.48\linewidth}
      \centering
      \includegraphics[width=0.5\linewidth]{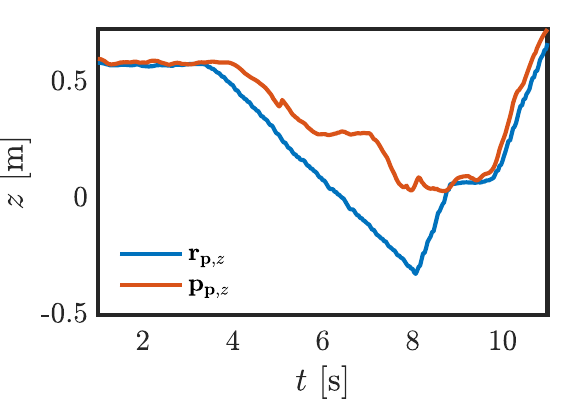}\hfill
      \includegraphics[width=0.5\linewidth]{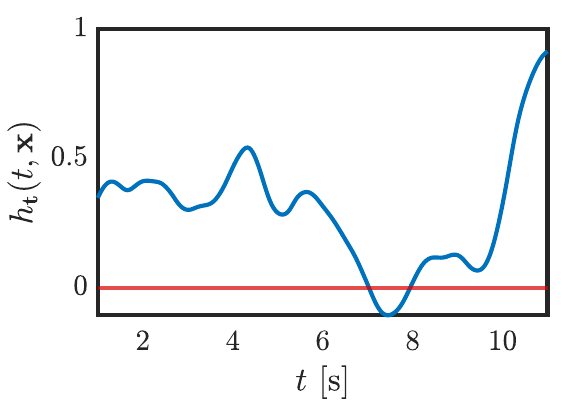}
    \end{minipage}
  }
  \hfill
  \subfloat[Experimental results with robust safety constraint.\label{fig:robust_payload_safety}]{
    \begin{minipage}{0.48\linewidth}
      \centering
      \includegraphics[width=0.5\linewidth]{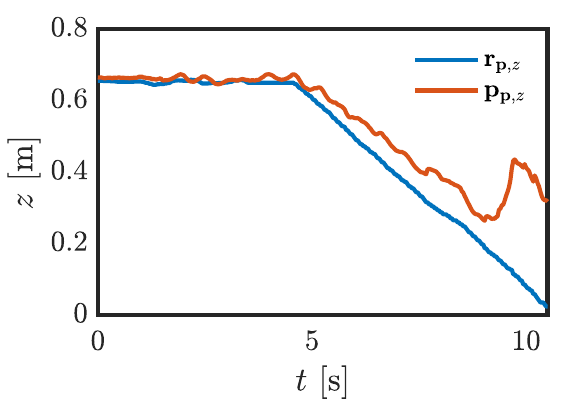}\hfill
      \includegraphics[width=0.5\linewidth]{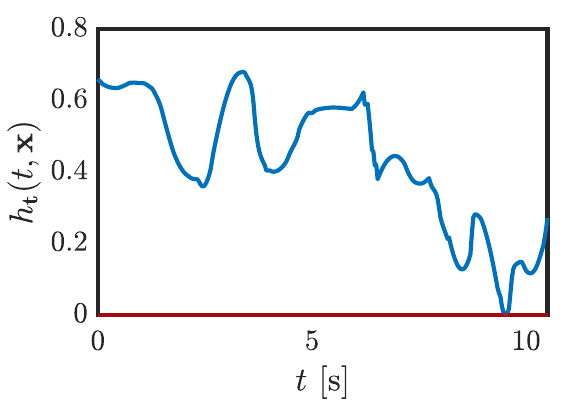}
    \end{minipage}
  }

  \caption{Comparison of nominal and robust safety configurations across hardware runs. Each subfigure contains the vertical payload tracking $(\mathbf{r}_{\mathbf{p},z}, \mathbf{p}_{\mathbf{p},z} \ \text{vs.\ time})$ on the left and the corresponding safety function $h_{\mb t}(t,\mathbf{x}(t))$ on the right. Safety is violated when the nominal safety function $h_{\mb t}(t, \mb x(t))$ is utilized in the MPC. Our robust MPC with R-ZOCBF constraint consistently preserves safety while staying close to the reference.
}
  \label{fig:hardware_payload_safety_panel}
   \vskip - 5mm
\end{figure*}

The hardware experiments aim to simulate with high fidelity a crane mounted on the deck of a ship at sea, as shown in Fig.~\ref{fig:compose}, Fig.~\ref{fig:hardware}, and Fig.~\ref{fig:crane}. Similar to the simulation, the hardware experiments focus on inserting a suspended payload into a bounded target aperture while avoiding static obstacles. Both the target and obstacles are mounted to the top plate of the Stewart platform, as if placed on the deck of a ship. The intended application also requires the payload not to hit the deck beyond the lip of the target receptacle. The payload consists of a PVC pipe with a diameter of 4 cm and a length of 60 cm. The target tube features an aperture of 5 cm and a cone aperture of 15 cm.

The crane system consists of a custom 8 ft carbon fiber boom mounted to a QuickSet MPT-50 Pan and Tilt Positioner that provides yaw and pitch motions. A weighted fishing line acts as the tether, which is spooled from the back of the crane and runs along the top of the boom. The crane is attached to the top plate of a Stewart platform that simulates periodic wave motion identical to that used in the simulation. 

All experiments are run using three "on-board" computers and one operator computer, connected to the operator controls. The "on-board" computers serve three roles: GNC, Perception-Payload-Estimation and Perception-Mapping. The "GNC" computer is an NVIDIA Jetson AGX Orin, where the MPC and safety mechanisms run. The Perception-Payload-Estimation computer is an Acer Predator Helios 18 laptop with an Intel i9-14900HX processor, 32GB of DDR5 RAM, and a GeForce RTX 4090 with 16GB of VRAM. Its task is to perform visual tracking and pose estimation of the payload using the mid-boom camera stream. Finally, the Perception-Mapping computer is an NVIDIA Jetson AGX Orin, which tracks the target pose and produces and updates a map of the ship deck for obstacle avoidance.

The perception system consists of two Intel RealSense D457 cameras. The first is mounted near the middle of the boom with the optical axis tilted 45$^\circ$ down from the axis of the boom. The camera is positioned vertically to ensure maximum visibility of the payload at various tether lengths. The second is mounted near the tip of the boom, pointing backwards, at an angle of 60$^\circ$ with respect to the boom axis. This camera enables consistent visual tracking of the target, as well as mapping of the "deck" of the ship for obstacle avoidance purposes. The cameras are configured to ${848 \!\times\! 480}$ resolution at 60 FPS and ${1280 \!\times\! 720}$ resolution at 30 FPS, respectively. The higher frame rate, synonymous with lower latency, of the mid-boom camera is crucial for enabling the control system to perform stabilization in the presence of external disturbances. Inertial measurements are provided by Vectornav VN-100 inertial measurement units attached to the crane base and mounted above the mid-boom camera. The hardware setup is shown in Fig.~\ref{fig:hardware}.

The perception pipeline consists of two main components: real-time object segmentation and object-specific pose estimation. The foundation of our perception pipeline is a video object segmentation (VOS) model. This model classifies the pixels based on their belonging to a desired object to track. The process is initialized by a human operator who provides an initial selection, e.g., via clicks and bounding boxes, on the objects of interest. The VOS model then propagates these segmentation masks to subsequent frames automatically.

For this task, we employ SAM2-Tiny \cite{ravi2024sam2} due to its high efficiency and state-of-the-art performance, which are critical for our low-latency requirements. The model provides a binary mask for each tracked object in every new frame. Along with depth images, these segmentation masks serve as the input for the pose estimation modules.

We determine the pose of the payload and target in two ways. First, the target pose is obtained by capturing a set of segmented point clouds in the first few frames, creating a smoothed model. Target pose estimation is achieved via iterative closest point (ICP) fits from newly detected segmented point clouds and the initial model. Payload pose estimation cannot reliably rely on ICP due to fast movements and varying observed aspect ratios of the cylinder. We therefore exploit symmetries of the cylinder and its projection on images. The main assumption is that the cylinder is "long" and that its side remains visible almost all the time. To obtain an estimate of the orientation of the cylinder, we first fit a line through the segmented mask and select a range of pixels in the vicinity of the line. We project the depths of the selected pixels into a corresponding point cloud, where we perform random sample consensus (RANSAC) line fitting, yielding the orientation of the cylinder.

With the control and perception systems in place, we assessed whether the proposed safety constraints could achieve safe insertion performance under dynamic base motion. The platform for the hardware experiment was subjected to the same lateral sinusoidal translation as the platform in simulation, ${\pm 0.05~\mathrm{m}}$ with one full cycle occurring every second. The desired position for the payload was then set at a point inside the target, and the system response was tracked. We evaluated safety and performance under two constraints: 
(i) a robust safety constraint—the R-ZOCBF condition defined in \eqref{eq:cbf_condition}, 
and (ii) a nominal safety constraint corresponding to the unmodified target safety function ${h_{\mb t}(t,\mb x) \!\geq\! 0}$ (i.e., no robustness margin) introduced in Section~\ref{sec:safe_function}. Fig.~\ref{fig:nominal_payload_safety} shows the failure of our system to maintain safety when commanded to enter the target area under disturbance. In contrast, Fig.~\ref{fig:robust_payload_safety} illustrates that under our R-ZOCBF constraint, the system maintains safety despite base disturbances and insertion commands that would otherwise be unsafe.

\section{Conclusion and Future Work} 
\label{sec:conc}
This work introduces the first experimental demonstration of a robust MPC framework for ship-mounted cranes that enforces safety via a robust zero-order control barrier function (R-ZOCBF) constraint. Our safe control method couples a smooth safety function for payload insertion with time-varying bounding boxes generated using onboard perception, and it uses a sampling-based online method to tune the R-ZOCBF margin. Multiple simulations and hardware trials achieved accurate placement of the payload while satisfying constraints in the presence of significant base excitation. We observe that online adaptation enables the system to automatically tune conservatism, balancing safety with the goal of successful payload transfer.

In addition to ship-mounted cranes, our approach applies more generally to robotic insertion, assembly, and manipulation under uncertainty. Future work will extend this approach to address challenges such as safe peg-in-hole assembly and spacecraft docking problems, incorporating contact dynamics. We also plan to formalize guarantees of stability and recursive feasibility, as well as integrate learning-based prediction of environmental disturbances.

\begin{spacing}{0.87}
\bibliographystyle{IEEEtran}
\bibliography{Bib/refs}
\end{spacing}

\end{spacing}

\end{document}